\documentclass[aps,prl,10pt,twocolumn,superscriptaddress,floatfix]{revtex4-2}
\pdfoutput=1

\usepackage{amsmath}
\usepackage{amssymb}
\usepackage{amsthm}
\usepackage{mathrsfs}
\usepackage{bm, dsfont}
\usepackage[usenames,dvipsnames]{color}
\usepackage{colortbl}
\usepackage[table]{xcolor}
\usepackage{enumitem}
\usepackage{graphicx}
\usepackage{natbib}
\usepackage[colorlinks=true,linkcolor=blue,citecolor=blue,urlcolor=blue]{hyperref}
\usepackage{hypcap}
\usepackage{verbatim}
\usepackage{psfrag}
\usepackage[normalem]{ulem}
\usepackage{physics}
\usepackage[english]{babel}
\usepackage{thmtools, thm-restate}

\declaretheorem{theorem}
\theoremstyle{definition}
\newtheorem{definition}{Definition}

\DeclareMathOperator{\conv}{conv}

\newcommand{\updownarrows}{\uparrow\!\downarrow}
\renewcommand{\paragraph}[1]{\addcontentsline{toc}{section}{#1}\emph{#1.}---}

\begin{document}

\title{Ruling out nonlinear modifications of quantum theory with contextuality}

\author{Ruben Campos Delgado}
\email{ruben.camposdelgado@itp.uni-hannover.de}
\author{Martin Plávala}
\email{martin.plavala@itp.uni-hannover.de}
\affiliation{Institut für Theoretische Physik, Leibniz Universität Hannover, 30167 Hannover, Germany}

\begin{abstract}
Nonlinear modifications of quantum theory are considered potential candidates for the theory of quantum gravity, with the intuitive argument that since Einstein field equations are nonlinear, quantum gravity should be nonlinear as well. Contextuality is a property of quantum systems that forbids the explanation of prepare-and-measure experiments in terms of classical hidden variable models with suitable properties. We show that some well-known nonlinear modifications of quantum mechanics, namely the Deutsch’s map, the Weinberg’s model, and the Schrödinger–Newton equation, map a contextual set of states to a non-contextual one. That is, the considered nonlinear modifications of quantum theory allow for the existence of classical hidden variable models for certain experimental setups. This enables us to design experiments that would rule out the considered nonlinear modifications of quantum theory by verifying that the system remains contextual, or, equivalently, our results highlight a mechanism how nonlinear modification of quantum theory may lead to weak wave function collapse and ultimately to the solution of the measurement problem.
\end{abstract}

\maketitle

\paragraph{Introduction}%
Nonlinear modifications of quantum theory have a long tradition dating back to de Broglie \cite{deBroglie1960}. From a fundamental point of view, it is possible that the theory of quantum gravity might involve nonlinear dynamics \cite{Karolyhazy:1966zz,DIOSI1984199,Diosi:1986nu,Penrose:1996cv,Penrose:1998qc}, and such modifications suggest experimental tests of quantum mechanics itself \cite{Weinberg62.485, WEINBERG1989336}. From a computational point of view, nonlinear dynamics have the potential to solve NP-complete problems in polynomial time \cite{Abrams3992}. However, all these endeavors face a problem when relativity is taken into account. In 1989 Gisin \cite{gisin1989stochastic,Gisin1990} showed that nonlinear quantum dynamics would lead to superluminal signaling. This effect was investigated repeatedly \cite{polchinski1991weinberg,lucke1999nonlocality,gruca2024correlations}, yielding both claims that no superluminal signaling implies linearity \cite{caticha1998consistency,svetlichny2005nonlinear,jordan2009quantum,bassi2015no} and nonlinear modifications of quantum theories that do not feature superluminal signaling \cite{gisin1995relevant,ferrero2004nonlinear,kent2005nonlinearity,zloshchastiev2010logarithmic,rembielinski2020nonlinear,kaplan2022causal,geszti2024nonlinear,Bielinska:2024trs}. Certain nonlinear effects in quantum theory were also proposed to be experimentally tested \cite{sahoo2022testing,diosi2025causality} and some of the experiments were also carried out \cite{donadi2021underground,polkovnikov2023experimental,broz2023test}.

Contextuality is a fundamental concept which captures the idea that measurements in quantum theory cannot be considered as revealing pre-existing values of a specific classical hidden variable. The original notion of contextuality \cite{Kochen:1967}, applying only to projective measurements, was extended by Spekkens to include unsharp measurements and eventually general operational theories \cite{Spekkens:2005vtt}. Contextuality was observed to be a resource for quantum computation \cite{howard2014contextuality}, state discrimination \cite{schmid2018contextual}, or useful for testing dimensions of quantum systems \cite{haakansson2025experimental}. Contextuality is also connected to other non-classical aspects of quantum theory such as entanglement \cite{plavala2024contextuality}, measurement incompatibility and steering \cite{tavakoli2020measurement,plavala2022incompatibility}, or Bell nonlocality \cite{wright2023invertible}. More recently, a complete algebraic characterization of contextuality has been proposed \cite{Frembs:2024oxm, Frembs:2025tyw}.

\begin{figure}
\centering
\includegraphics[width=0.66666\linewidth]{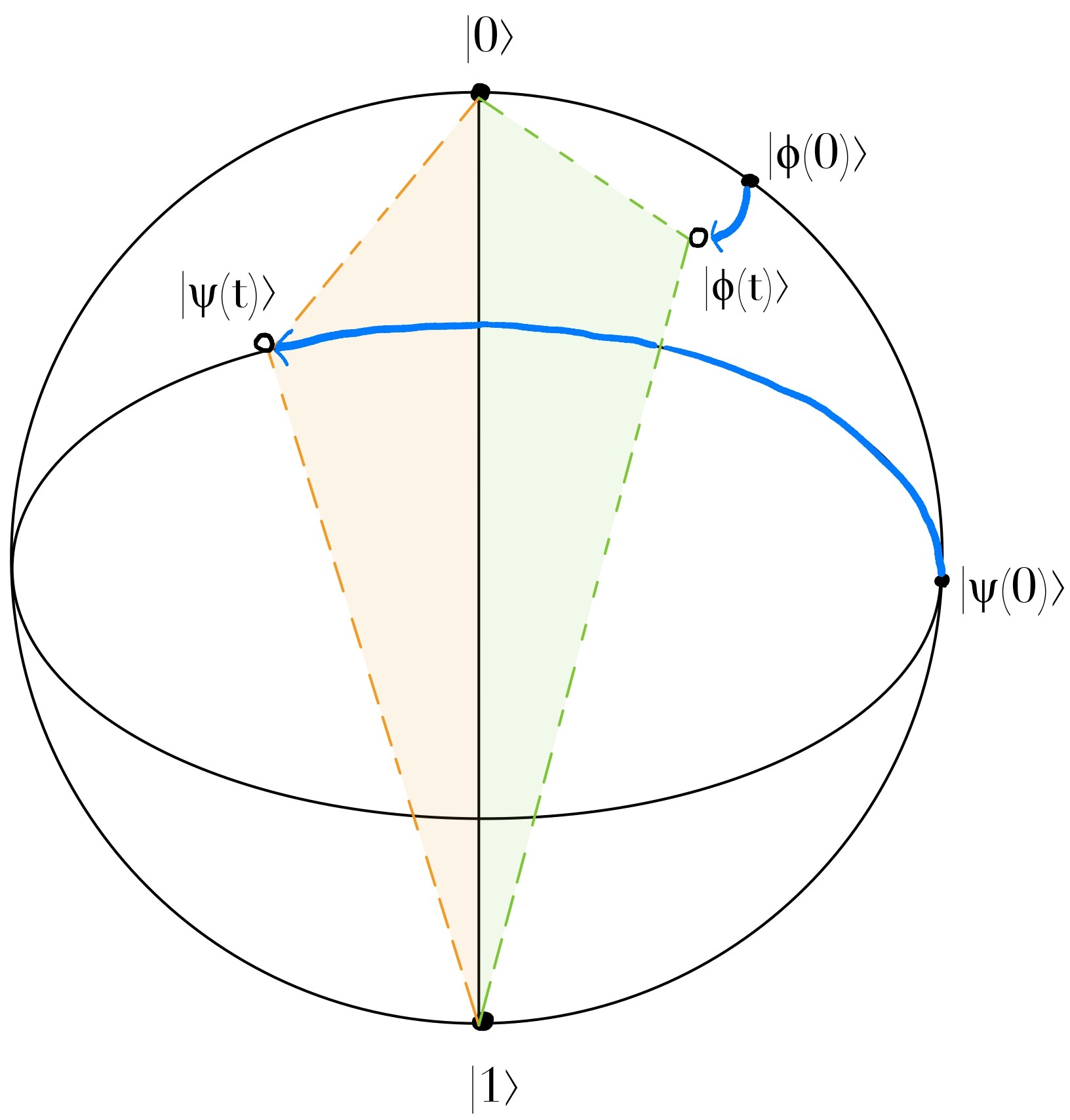}
\caption{Example of an action of a nonlinear map acting on a qubit. The black dots are the initial states. $\ket{0}$ and $\ket{1}$ are invariant under the nonlinear time evolution, the other states move along the blue arrows on the surface of the Bloch sphere. The nonlinear map displaces states which were initially on the same plane to different planes. Thus contextual quantum states become non-contextual.}
\label{fig: fig1}
\end{figure}

In this paper we show that Spekkens contextuality can be used to rule out proposed nonlinear modifications of quantum theory in a prepare-transform-measure type experiments that involves only single quantum systems. Using only single quantum systems avoids the action of the nonlinear dynamics on a bipartite systems and thus we do not consider any problems caused by superliminal signaling. We focus on three well-known, physically motivated examples: the Deutsch's cloning map \cite{Deutsch:1991nm} based on closed time-like curves, the Weinberg's model \cite{Weinberg62.485, WEINBERG1989336} and the Schrödinger-Newton equation \cite{Penrose:1996cv, Penrose:1998qc}. We consider an initial set of states that is contextual, and we show that the considered models predict that the initial set becomes non-contextual, thereby all violations of non-contextuality inequalities \cite{Spekkens2009, Kunjwal2018, Mazurek2016, Pusey2018} vanish. Since the Schrödinger-Newton equation lacks analytical solutions even in the simplest cases, we use the iterative approach developed by Grossardt \cite{Grossardt:2021ytg} and show that the result holds in the first iteration. We prove the following theorem in order to decide whether a set of states is contextual or not: a set of pure states is non-contextual if and only if their density matrices are linearly independent operators. Then for all of the considered nonlinear models we show that the nonlinear dynamics maps linearly dependent, hence contextual, density matrices to linearly independent, hence non-contextual, ones; see Fig.~\ref{fig: fig1} for a demonstration. This feature of nonlinear dynamics can be tested experimentally. If the states remain contextual for all the duration of the experiment, and thus they always violate the non-contextuality inequalities, then this dismisses the aforementioned nonlinear models.

The paper is structured as follows: we give a proper formal definition of a non-contextual set of states and we prove a necessary and sufficient condition for a set of pure states to be contextual. Then we explicitly study the evolution of an initial contextual set of states under nonlinear dynamics in three specific cases: Deutsch's cloning map \cite{Deutsch:1991nm}, the Weinberg's model \cite{Weinberg62.485, WEINBERG1989336} and the Schrödinger-Newton equation \cite{Penrose:1996cv, Penrose:1998qc}. We show that all of them map an initial contextual set of states to a non-contextual one. 

\paragraph{Spekkens contextuality}%
The Bell-Kochen-Specker theorem \cite{Bell:1966, Kochen:1967, Budroni:2021rmt} rules out the existence of a local non-contextual hidden variable model of quantum theory. Such a non-contextual model is one wherein the measurement outcome that occurs for a particular set of values of the hidden variables depends only on the Hermitian operator associated with the measurement and not on which Hermitian operators are measured simultaneously with it. Throughout our work we will use the notion of contextuality as defined by Spekkens \cite{Spekkens:2005vtt}. We will always assume that the set of possible measurements includes at least all projective measurements (PVMs), this enables us to define the notion of non-contextuality and contextuality of set of states.

We will use $\mathcal{H}$ to denote a finite-dimensional complex Hilbert space and $\mathcal{D}(\mathcal{H})$ the set of density matrices, that is, positive semidefinite operators with unit trace. A projection-valued measure (PVM) $\{\Pi_k\}_k$ is a collection of projectors, $\Pi_k = \Pi_k^2$, that sums to the indentity operator, $\sum_k \Pi_k = \mathds{1}$.
\begin{definition}\label{def1}
Let $I$ be a set of indices, the set $\{\rho_i\}_{i\in I} \subseteq \mathcal{D}(\mathcal{H})$ is said to be non-contextual if, for any PVM $\{\Pi_k\}_k$, the probability of observing outcome $k$ can be expressed as
\begin{equation}\label{eq1}
\Tr(\rho_i \Pi_k) = \sum_{\lambda\in \Lambda} \Tr(\rho_i G_{\lambda})\Tr(\sigma_\lambda \Pi_k) \quad \forall i\in I,\,\, \forall \Pi_k
\end{equation}
where $\Lambda$ is an index set, $\Tr(\rho_i G_{\lambda})\geq 0$, $\Tr(\sigma_{\lambda} \Pi_k)\geq 0$, and $\Tr(\sigma_\lambda) = 1$ for all $\lambda \in \Lambda$ and $\sum_{\lambda} \Tr\left(\rho_i G_{\lambda}\right) = 1$ for all $i \in I$.
\end{definition}
Since we consider all PVMs $\{\Pi_k\}_k$, the condition $\Tr(\sigma_{\lambda} \Pi_k)\geq 0$ implies that $\{\sigma_{\lambda}\}_{\lambda \in \Lambda}$ are quantum states. If Eq.\eqref{eq1} does not hold, the set $\{\rho_i\}_{i\in I}$ is called contextual. The index $\lambda$ is usually referred to as hidden variable.

Note that if $\{\rho_i\}_{i\in I} \subseteq \mathcal{D}(\mathcal{H})$ is non-contextual, then also $\conv(\{\rho_i\}_{i\in I})$, the set of all convex combinations of $\rho_i$, is non-contextual. Hence one can, without the loss of generality, restrict to the extreme points of $\{\rho_i\}_{i\in I}$.

We can immediately formulate the following theorem connecting linear independence of the density matrices with contextuality. Special cases of this result have been observed before \cite{beltrametti1995classical}.
\begin{theorem} \label{th1}
Let $\mathcal{H}$ be a Hilbert space, let $\mathcal{D}(\mathcal{H})$ be the space of states in $\mathcal{H}$ and let $I$ be a set of indices. If the states  $\{\rho_i\}_{i\in I} \subseteq \mathcal{D}(\mathcal{H}) $ are linearly independent, then they are non-contextual. 
\end{theorem}
\begin{proof}
Since the states are linearly independent, there exists a dual basis $\{F_j\}_{j\in I }\subseteq \mathcal{D}(\mathcal{H})$ such that $\Tr(\rho_i F_j)=\delta_{ij}$ for all $i,j \in I$. Therefore, $\rho_i=\sum_{j}\delta_{ij}\rho_j=\sum_{j} \Tr(\rho_i F_j)\rho_j$ for all $i\in I$. For any POVM $\{M_k\}_k$, it then follows that $\Tr(\rho_i M_k)=\sum_j \Tr(\rho_i F_j)\Tr(\rho_j M_k)$ for all $i\in I$.
\end{proof}

In general, the converse is not true as sufficiently noisy quantum systems will be non-contextual \cite{marvian2020inaccessible}. Nevertheless, if the states are pure, then linear independence becomes a necessary and sufficient condition for non-contextuality, as partially observed in \cite{schmid2018contextual,Rossi:2022jsr} in a specific case of the set of four pure states of a qubit.
\begin{restatable}{theorem}{thMain} \label{th3}
Let $\mathcal{H}$ be a Hilbert space, let $\mathcal{D}(\mathcal{H})$ be the space of states in $\mathcal{H}$, let $I$ be a set of indices, and let $\{\rho_i\}_{i\in I} \subseteq \mathcal{D}(\mathcal{H})$ be a set of pure states. The set $\{\rho_i\}_{i\in I}$ is non-contextual if and only if the density matrices of the states $\rho_i$ are linearly independent.
\end{restatable}
We relegate the proof to the Appendix.
We now have all the necessary tools to show that some nonlinear modifications of quantum theory transform a set of linearly dependent density matrices to a set of linearly independent density matrices. Theorem~\ref{th3} then implies that a contextual set of states is transformed to a non-contextual set of states. One would be tempted to conclude that any nonlinear modification of quantum theory exhibits this effect, but we construct a counter-example in the Appendix.

\paragraph{Deutsch's cloning map}%
The first case we investigate is the nonlinear map that perfectly clones a pure quantum state by using closed time-like curves \cite{Deutsch:1991nm,brun2013quantum}:
\begin{equation}
    \phi(\rho)= \rho \otimes \rho.
\end{equation}
Let us consider the initial pure states at the antipodes of the Bloch sphere, $\{\dyad{0}, \dyad{1}, \dyad{+}, \dyad{-}\}$, where $\ket{0}, \ket{1}$ represent the standard computational basis and $\ket{\pm} = \frac{\ket{0} \pm \ket{1}}{\sqrt{2}}$. The density matrices are linearly dependent since $\dyad{0} + \dyad{1} = \dyad{+} + \dyad{-}$. After applying the map, the tensor product states become
linearly independent: one can verify this by explicitly writing down the respective matrices.

\paragraph{Weinberg's model}%
The second case we investigate is the Weinberg's model of nonlinear quantum mechanics  \cite{Weinberg62.485, WEINBERG1989336}. The Hamiltonian is a real homogeneous nonbilinear function $h(\psi,\psi^*)$ of degree one, and the evolution of a $k$-dimensional state vector is governed by the equation 
\begin{equation}\label{eq:weinberg}
\dfrac{\partial \psi_k}{\partial t} = -i \frac{\partial h}{\partial \psi^*_k},
\end{equation}
where $\psi_k$ is the $k$-th component of the state vector, i.e., $\psi_k = \bra{k} \ket{\psi}$ where $\ket{k}$ is the respective element of the computational basis.

We will now restrict to the case of a single qubit. A generic pure state of a qubit can described in terms of the angles $\vartheta \in [0,\pi)$ and $\varphi \in [0,2\pi)$ as
\begin{equation}
\begin{split}
\ket{\psi} =
\begin{pmatrix} \cos(\vartheta) \\ \sin(\vartheta) e^{i \varphi}
\end{pmatrix}.
\end{split}
\end{equation}
For a qubit, one can always rewrite the Hamiltonian $h(\psi,\psi^*)$ as $h = \bar{h}(\vartheta)$ in terms of the angle $\vartheta$. The solution of Eq. \eqref{eq:weinberg} is \cite{Abrams3992}
\begin{equation}\label{eq: new_state}
\ket{\psi(t)} = \begin{pmatrix} \cos(\vartheta) \\ \sin(\vartheta) e^{i (\varphi - \omega(\vartheta) t)} \end{pmatrix}
\end{equation}
where
\begin{equation}
\omega(\vartheta) = \frac{\bar{h}'(\vartheta)}{2\sin\theta\cos\vartheta}.
\end{equation}
The full derivation is in the Appendix. The time evolution rotates the states around the Bloch sphere, but the angular speed at which the states move depends on $\vartheta$. In order to find a set of states that is initially contextual but is transformed to a non-contextual set via the nonlinear time evolution, we  choose the initial states to lie on a fixed meridian of the Bloch sphere as then their density matrices are linearly dependent. The nonlinear time evolution will then move the states to different meridians, making them linearly independent and hence non-contextual, as also demonstrated in Fig.~\ref{fig: fig1}.

Let us then consider four pure states $\lvert \psi_i\rangle$, $i=1,\dots 4$, described by angles $\vartheta_i$, where we set $\varphi_i = 0$ for all $i$. Let $\vartheta_3 = 0$ and $\vartheta_4 = \pi/2$ so that $\ket{\psi_3} = \ket{0}$ and $\ket{\psi_4} = \ket{1}$ are the computational basis states that are invariant under the nonlinear time evolution. Without loss of generality, we can assume that there are $\vartheta_1, \vartheta_2 \in (0, \frac{\pi}{2})$ such that $\omega(\vartheta_1) \neq \omega(\vartheta_2)$. In fact, if $\omega(\vartheta)$ were constant for all $\vartheta \in (0, \frac{\pi}{2})$, then the time evolution would coincide with rotation of the Bloch sphere along the axis, which is generated by linear dynamics.

We now show that the density matrices of the initial states evolve into linearly-independent density matrices. 
While checking the linear dependence/independence, we are faced with the following set of equations, for some real numbers $\alpha_i$:
\begin{equation}\label{eq: conditions}
\begin{split}
\sum_{i=1}^4 \alpha_i \cos^2\vartheta_i = 0 \\
\sum_{i=1}^4 \alpha_i \sin^2\vartheta_i = 0 \\
\sum_{i=1}^4 \alpha_i \cos\vartheta_i\sin\vartheta_i e^{\pm i \omega(\vartheta_i)t} = 0.
\end{split}
\end{equation}
Now,  $\cos\theta_1\sin\theta_1\neq 0$ so $\alpha_1$ can be written in two different ways, depending on the positive or negative sign in the exponent:
\begin{equation}
    \alpha_1 = -\frac{1}{\cos\theta_1\sin\theta_1}\sum_{i=2}^{4}\cos\theta_i\sin\theta_i 
    e^{\pm i \left[ \omega(\vartheta_i)-\omega(\vartheta_1) \right]t }. 
\end{equation}
Equating both expressions and removing the terms corresponding to $\theta=0$ and $\theta=\pi/2$, we get 
\begin{equation}\label{eq: last}
    \alpha_2 \cos\theta_2\sin\theta_2 \sinh \left[ \omega(\vartheta_2)-\omega(\vartheta_1)\right]t = 0.
\end{equation}
We argued earlier that $\omega(\vartheta_1)\neq \omega(\vartheta_2)$. Therefore, Eq. \eqref{eq: last} automatically implies $\alpha_2 = 0$ and by Eq. \eqref{eq: conditions} we also have $\alpha_2=\alpha_3=\alpha_4=0$.
The conclusion is that the original linearly dependent states become linearly independent after evolving according to the Weinberg's map.

\paragraph{Schrödinger-Newton equation}%
Our last example is the Schrödinger-Newton (SN) equation proposed by Diosi and Penrose to explain the gravity-driven collapse of the quantum mechanical wave function. \cite{Diosi:1986nu, Penrose:1996cv, Penrose:1998qc}. We consider the set-up proposed in \cite{Grossardt:2021ytg}, where a spin-$\frac{1}{2}$ particle of mass $m$ is sent through a non-uniform magnetic field directed along the $z$ axis. The state vector of the particle undergoes the transformation
\begin{equation}
\begin{pmatrix}
\alpha \\ \beta 
\end{pmatrix}
\mapsto
\begin{pmatrix}
\alpha \psi_{\uparrow}\\ \beta \psi_{\downarrow}
\end{pmatrix},
\end{equation}
where $\psi_{\updownarrows}$ obey the SN equation:
\begin{equation}\label{eq: SN}
i\hbar \frac{\partial \psi_{\updownarrows}}{\partial t} = \left(-\frac{\hbar^2}{2m}\nabla^2 -\mathbf{F}_{\updownarrows}\cdot \mathbf{r}+|\alpha|^2U_{\uparrow}+|\beta|^2U_{\downarrow}\right)\psi_{\updownarrows},
\end{equation}
with $\mathbf{F}$ representing the compound effects of the Newtonian gravitational and magnetic forces, and
\begin{equation}
\begin{split}
U_{\updownarrows}(t,\mathbf{r})=\int d^3 \mathbf{r}' \, \lvert \psi_{\updownarrows}(t,\mathbf{r})\rvert^2I_{\rho}(\mathbf{r}-\mathbf{r}'),\\
I_{\rho}(\mathbf{d})=-G_{N}\int d^3 \mathbf{r} d^3\mathbf{r}'\, \frac{\rho(\mathbf{r})\rho(\mathbf{r}')}{\lvert \mathbf{r}-\mathbf{r}'+\mathbf{d} \rvert}.
\end{split}
\end{equation}
Eq. \eqref{eq: SN} can be solved recursively, by first setting $U_{\updownarrows}$ to zero. At each iteration, $U_{\updownarrows}$ is calculated from the solution at the previous step.
The solution factorizes as $\psi_{\updownarrows}(t,\mathbf{r})=e^{\tau_{\updownarrows}(t,\mathbf{r})}\chi_{\updownarrows}(t,\mathbf{r})$, where $\tau_{\updownarrows}$ is a phase which depends solely on the classical trajectory, and is irrelevant for our discussion. At first order, $\chi_{\updownarrows}(t,\mathbf{r})$ satisfies the equation \cite{Grossardt:2021ytg}
\begin{equation}
    i\hbar \frac{\partial \chi^{(1)}_{\updownarrows}}{\partial t}(t,\mathbf{r}) = \left(-\frac{\hbar^2}{2m}\nabla^2 +U^{(0)}_{\updownarrows}(t,\mathbf{r})\right)\chi^{(1)}_{\updownarrows},
\end{equation}
where
\begin{equation}
\begin{split}
U^{(0)}_{\uparrow}(t,\mathbf{r}) &= |\alpha|^2\Tilde{U}(t,\mathcal{r})+|\beta|^2\Tilde{U}(t,\lvert \mathbf{r}+\Delta\mathbf{u}(t)\rvert)\\
U^{(0)}_{\downarrow}(t,\mathbf{r}) &= |\alpha|^2\Tilde{U}(t,\lvert \mathbf{r}-\Delta\mathbf{u}(t)\rvert)+|\beta|^2\Tilde{U}(t,\mathbf{r}),\\
\tilde{U}(t,\mathbf{r}) &= \int d^3 \mathbf{r}'\, \lvert \chi^{(0)}(t,\mathbf{r}')\rvert^2 I_{\rho}(\mathbf{r}-\mathbf{r}'),\\
\Delta \mathbf{u} &= \mathbf{u}_{\uparrow}-\mathbf{u}_{\downarrow}
\end{split}
\end{equation}
and where $\mathbf{u}_{\updownarrows}$ solves the equations of motion
\begin{equation}
\ddot{\mathbf{u}}_{\updownarrows} = \frac{\mathbf{F}_{\updownarrows}(t)}{m}.
\end{equation}
The $z$-component of the solution at first order is \cite{Grossardt:2021ytg} 
\begin{equation}
\chi^{(1)}_{\updownarrows}(t,z)=(2\pi A_{\updownarrows})^{-1/4}\exp\left[-\frac{\left(z-\langle z\rangle_{\updownarrows}\right)^2}{4A_{\updownarrows}}\right]e^{i\phi_{\updownarrows}(t,z)},
\end{equation}
where 
\begin{equation}
\phi_{\updownarrows} = \frac{\left(z-\langle z\rangle_{\updownarrows}\right)^2B_{\updownarrows}}{4A_{\updownarrows}\hbar}+\frac{\langle p \rangle_{\updownarrows}z+f_{\updownarrows}}{\hbar}
\end{equation}
and the coefficients are defined as follows
\begin{equation}
\begin{split}
A_{\updownarrows} &= \langle z^2 \rangle_{\updownarrows}-\langle z \rangle^2_{\updownarrows},\\
B_{\updownarrows} &= \langle zp+pz\rangle_{\updownarrows}-2\langle z \rangle_{\updownarrows}\langle p \rangle_{\updownarrows},\\
f_{\updownarrows} &= -\frac{\langle z \rangle_{\updownarrows} \langle p \rangle_{\updownarrows}}{2}-\frac{\hbar^2}{4m}\int_0^t \frac{dt'}{A_{\updownarrows}(t')} \\
&\hphantom{=}\; -\int_0^tdt'\, \left( U^{(0)}_{\updownarrows|_{z=0}}+\frac{1}{2}\partial_zU^{(0)}_{\updownarrows | _{z=0}}\langle z \rangle_{\updownarrows}(t') \right).
\end{split}
\end{equation}
The two phases depend on $\delta = |\alpha|^2-|\beta|^2$ and they are different from each other \cite{Grossardt:2021ytg}. Hence, the evolution of states resembles the Weinberg's dynamics of the previous example, see Eq. \eqref{eq: new_state} and Fig.~\ref{fig: fig1}. Accordingly, a contextual set of states is mapped to a non-contextual one.

\paragraph{Conclusions}
We have shown that nonlinear dynamics has the ability to map contextual sets of quantum states to non-contextual sets. Our results provide a blueprint in how to design fundamental experiments similar to the measurement of violations of Bell inequalities that would rule out various nonlinear modification of quantum theory based on non-existence of a non-contextual hidden variable model. For the Deutsch's cloning map and for the Weinberg’s model the result is immediate and explicit sets of states that become non-contextual due to the nonlinear transformation were indentified. For the Schrödinger–Newton equation we were restricted by the lack of analytic solutions of the nonlinear differential equation, but we showed that in the first iteration of the solutions constructed by Grossardt \cite{Grossardt:2021ytg} the same effect occurs and the nonlinear time evolution causes a wave function-dependent phase shift.

Our results can be directly used to construct experimental test of the nonlinear modifications of quantum theory. To do so, one needs to identify a contextual set of quantum states that becomes non-contextual due to the nonlinear dynamics. Since the initial set of states is contextual, one can use available software to find a non-contextuality inequality it violates \cite{selby2024linear}, and then the violation can be experimentally verified \cite{Mazurek2016,Pusey2018}. After the nonlinear time evolution is assumed to have occurred, a non-contextuality inequality will be measured again. If the nonlinear time evolution occurred, then the final set of states will be non-contextual and the non-contextuality inequality will not be violated. We note that the final states not violating any non-contextuality inequality is not a definitive proof that the nonlinear dynamics occured, since the final set of states can also become non-contextual due to decoherence \cite{marvian2020inaccessible}. But if the final states violate some non-contextuality inequality, the nonlinear time evolution is experimentally ruled out.

Our result can be interpreted to argue that while nonlinear dynamics do not have to cause explicit wave function collapse, we demonstrated that they have the ability to enable non-contextual hidden variable description of certain sets of states. That is, our results show that nonlinear dynamics may cause a weak wave function collapse where the wave function does not collapse to a classical stochastic state, but a description of sets of states in terms of a hidden variable model becomes possible. Thus we provide a generalization of previous results investigating classical aspects of nonlinear quantum theory \cite{haag1978comments,bugajski1991nonlinear,sahoo2023emergence}. The weak wave function collapse can also be a potential foundational way to address some known obstacles in collapse theories \cite{diosi2025healthier}.

The remaining question is how to experimentally test concrete nonlinear modifications of quantum theory. Unlike recent experiments ruling out the Diósi-Penrose collapse model \cite{donadi2021underground}, the experiments employing contextuality would require direct measurement of the quantum systems. Here we are facing a theoretical challenge: our results as stated apply only to discrete systems described by finite-dimensional Hilbert spaces. In order to properly test various nonlinear modifications of quantum theory, we will need to develop a continuous-variable generalization of contextuality. Given such generalization, it should be possible to generalize theorems \ref{th1} and \ref{th3}, with the caveat that linear independence of the density matrices will likely need to be replaced by a more suitable notion pertaining to density matrices on infinite-dimensional Hilbert spaces, or in general to states on operator algebras. With these results in hand, one can look for atom interferometric or other setups \cite{aspelmeyer2022zeh} where the sought effect could be observed, in correspondence with previously-proposed setups \cite{sahoo2022testing}. We leave these tasks to future endeavors.

\begin{acknowledgments}
\paragraph{Acknowledgments}
We are thankful to Shawn Skelton for useful discussion.
MP acknowledges support from the Niedersächsisches Ministerium für Wissenschaft und Kultur.
\end{acknowledgments}

\bibliography{citations}

\onecolumngrid
\appendix

\section{Proof of Theorem~\ref{th3}} \label{appendix:proof-th3}
\thMain*
\begin{proof}
We already argued in Theorem~\ref{th1} that linear independence of $\{\rho_i\}_{i\in I}$ implies non-contextuality. To prove the other implication, we show that if the states are linearly dependent, then they are contextual. Suppose, for the sake of contradiction, that the states $\{\rho_i\}_{i\in I}$ are linearly dependent and non-contextual.  Then, according to Definition \ref{def1}, there exist states $\{\sigma_{\lambda}\}_{\lambda}$ and a pseudo-POVM $\{G_{\lambda}\}_{\lambda}$ such that $\Tr(\rho_i M_k) = \sum_{\lambda} \Tr(\rho_i G_{\lambda})\Tr(\sigma_\lambda M_k)$ for all $i \in I$ and for all $M_k$, with $\Tr(\rho_i G_{\lambda})\geq 0$ for all $\lambda$. By the linearity of trace,
\begin{equation}
\Tr(\rho_i M_k) = \sum_{\lambda} \Tr \left[\sigma_{\lambda}\Tr(\rho_i G_{\lambda})M_k\right] = \Tr\left[\sum_{\lambda} \sigma_{\lambda}\Tr(\rho_i G_{\lambda})M_k\right].
\end{equation}
Since for arbitrary states $\omega$ and $\omega'$, the condition $\Tr(\omega M_k) = \Tr(\omega' M_k)$ for all $M_k$ implies $\omega = \omega'$. Hence
\begin{equation}
\rho_i = \sum_{\lambda} \Tr(\rho_i G_{\lambda})\sigma_{\lambda}.
\end{equation}
Let us now define the map
\begin{equation}
\Theta(\rho_i) := \sum_{\lambda}\Tr\left(\rho_i G_{\lambda}\right)\sigma_{\lambda}\otimes \sigma_{\lambda}.
\end{equation}
The partial trace over either subsystem is 
\begin{equation}
\Tr_A\left[\Theta(\rho_i)\right] = \Tr_B\left[\Theta(\rho_i)\right] = \sum_{\lambda}\Tr(\rho_i G_{\lambda})\sigma_{\lambda} = \rho_i
\end{equation} 
and since $\rho_i$ is pure, by monogamy of entanglement \cite{coffman2000distributed} we conclude that
\begin{equation}
\Theta(\rho_i) = \rho_i \otimes \rho_i, \quad \forall i \in I.
\end{equation}
Furthermore, since the states are linearly dependent, there exist real coefficients $\{\alpha_i\}_{i\in I}$, not all zero, such that $\sum_{i} \alpha_i \rho_i=0$. Let $\alpha_k\neq 0$ be one of such coefficients. The $k$-th state is then
\begin{equation}
\rho_k = -\sum_{i\neq k} \frac{\alpha_i}{\alpha_k}\rho_i.    
\end{equation}
Without loss of generality, we can assume that $\rho_k$ was the cause of the linear dependence, in such a way that when this is removed the remaining states are linearly independent. That is, without the loss of generality, we can assume that the set of states $\{ \rho_i : i \in I, i \neq k, \alpha_i \neq 0 \}$ is linearly independent. The effect of the map $\Theta$ on $\rho_k$ can be expressed in two different ways:
\begin{equation}
\Theta\left(\rho_k\right) = -\sum_{i\neq k} \frac{\alpha_i}{\alpha_k}\Theta(\rho_i) = -\sum_{i\neq k} \frac{\alpha_i}{\alpha_k} \rho_i \otimes \rho_i = -\sum_{i\neq k}\sum_{j \neq k}\frac{\alpha_i \delta_{ij}}{\alpha_k}\rho_i\otimes \rho_j
\end{equation}
and
\begin{equation}
\Theta(\rho_k)=\rho_k\otimes \rho_k = \sum_{i\neq k}\sum_{j \neq k } \frac{\alpha_i \alpha_j}{\alpha^2_k}\rho_i \otimes \rho_j.
\end{equation}
Comparing both expressions we get
\begin{equation}\label{coeff}
 \sum_{i\neq k }\sum_{j\neq k } \left(\frac{\alpha_i \delta_{ij}}{\alpha_k}+\frac{\alpha_i \alpha_j}{\alpha^2_k}\right)\rho_i\otimes \rho_j = 0.
\end{equation}
We argued before that $\{\rho_i\}_{i\neq k}$ is a set of linearly independent states. Therefore, also the tensor products $\rho_i \otimes \rho_j$ are linearly independent and \eqref{coeff} implies
\begin{equation}
\alpha_i \delta_ {ij} + \dfrac{\alpha_i\alpha_j}{\alpha_k} = 0, \quad \forall i,j \neq k.
\end{equation}
If $i=j$ we get $\alpha_i = 0$ or $\alpha_i = - \alpha_k \neq 0$ for all $i\neq k$. If $i \neq j$, then $\alpha_i \alpha_j = 0$. There are only two possible scenarios which are compatible with both conditions: either all coefficients are zero, or there exists exactly one non-vanishing coefficient. The former case implies $\rho_k = 0$, which is not a quantum state. For the latter, let $k'$ be the index corresponding to the non-vanishing coefficient. Then $\rho_k = \rho_{k'}$, i.e. two states are equal, which implies that the set of states we considered is in fact linearly independent.
\end{proof}

\section{Counter-example: nonlinear transformation that break contextuality}
Here we show that there exists nonlinear maps that do not map contextual sets of pure states to non-contextual sets of pure states. While the example that we provide is a valid nonlinear map, it is rather artificial.

For a single qubit, $\dim(\mathcal{H}) = 2$, the set of distinct pure states $\{ \rho_i \}_I$ is linearly independent hence non-contextual, whenever $|I| \leq 3$, that is, if the index set $I$ contains at most 3 indexes. If $|I| \geq 5$ then $\{ \rho_i \}_I$ must be linearly dependent hence contextual, the argument is by the pigeonhole principle since the space of hermitian operators has dimension $(\dim(\mathcal{H}))^2 = 4$. For $|I| = 4$ the distinct pure states $\{ \rho_i \}_I$ may be linearly dependent hence contextual if they lie in a plane, but they also may be linearly independent hence non-contextual.

Now consider a map $\Phi: \mathcal{D}(\mathcal{H}) \to \mathcal{D}(\mathcal{H})$ that maps the Bloch sphere to the disc in the $XZ$-plane, that is, to the set $\{ \rho \in \mathcal{D}(\mathcal{H}) : \Tr(\rho Y) = 0 \}$ where $Y$ is the corresponding Pauli operator. The map $\Phi$ can be chosen such that it maps pure states to pure states and that it is bijective, since both the Bloch sphere and the disc in the $XZ$-plane have the same cardinality.

We conclude the construction by investigating contextuality of pure states in the $XZ$-plane: let $\{ \rho_i \}_I \subset \{ \rho \in \mathcal{D}(\mathcal{H}) : \Tr(\rho Y) = 0 \}$ be set of pure distinct states in the $XZ$-plane, then they are linearly independent hence non-contextual if $|I| \leq 3$ and linearly dependent hence contextual if $|I| \geq 4$. It thus follows that the nonlinear map $\Phi$ maps contextual sets to contextual sets, and some non-contextual sets to non-contextual sets (for $|I| \leq 3$) while other non-contextual sets to contextual sets ($|I| = 4$).

\section{Weinberg's model for a single qubit}
We explicitly derive Eq. \eqref{eq: new_state}. The idea is to consider the normalization of the quantum state as a free parameter, i.e. we define the Hamiltonian as $h=n\bar{h}(a)$, where
\begin{equation}
\begin{gathered}
    n = \lvert \psi_1 \rvert^2 + \lvert \psi_2 \rvert^2, \\
    a = \frac{\lvert \psi_2\rvert^2}{n}=\frac{\sin^2\vartheta}{n}.
\end{gathered}
\end{equation}
The Weinberg equations \eqref{eq:weinberg} become 
\begin{equation}
\begin{gathered}
    i\frac{\partial \psi_1}{\partial t}=\frac{\partial h}{\partial \psi^*_1} = \psi_1 \bar{h}(a)+n\frac{\partial \bar{h}(a)}{\partial a}\frac{\partial a}{\partial \psi_1}=\psi_1\left[\bar{h}(a)-a\bar{h}'(a)\right], \\
    i\frac{\partial \psi_2}{\partial t}=\frac{\partial h}{\partial \psi^*_2} = \psi_2 \bar{h}(a)+n\frac{\partial \bar{h}(a)}{\partial a}\frac{\partial a}{\partial \psi_2}=\psi_2\left[\bar{h}(a)+(1-a)\bar{h}'(a)\right],
\end{gathered}
\end{equation}
with $\bar{h}'(a)\equiv \frac{\partial \bar{h}(a)}{\partial a}$. The solution is
\begin{equation}
    \psi_k = c_ke^{-i\omega_k(a)t}, \hspace{4mm} k=1,2
\end{equation}
where
\begin{equation}
\begin{gathered}
    \omega_1 = \bar{h}(a)-a\bar{h}'(a),\\
    \omega_2 = \bar{h}(a)+(1-a)\bar{h}'(a).
\end{gathered}
\end{equation}
The phase in the first component can be absorbed in the second one, resulting in the expression \eqref{eq: new_state} with 
\begin{equation}
    \omega(\vartheta)=-\omega_1+\omega_2=\bar{h}'(a)=\frac{\partial \vartheta}{\partial a}\bar{h}'(\vartheta)=\frac{1}{2\sin\vartheta\cos\vartheta}\bar{h}'(\vartheta).
\end{equation}

\end{document}